\documentclass[a4paper,USenglish]{lipics}
\usepackage{stylesheet}

\title{Solving $k$-SUM using few linear queries}
\titlerunning{Solving $k$-SUM using few linear queries} 

\author[1]{Jean Cardinal\thanks{Supported by the ``Action de Recherche Concert\'ee'' (ARC) COPHYMA, convention number 4.110.H.000023.}}
\author[2]{John Iacono\thanks{Research partially completed while on on sabbatical at the
Algorithms Research Group of the
D\'{e}partement d'Informatique at the
Universit\'{e} Libre de Bruxelles with support from
a Fulbright Research Fellowship,
the Fonds de la Recherche Scientifique --- FNRS,
and NSF grants CNS-1229185, CCF-1319648, and CCF-1533564.}}
\author[1]{Aur\'elien Ooms\thanks{Supported by the Fund for Research Training in Industry and Agriculture (FRIA).}}
\affil[1]{Universit\'e libre de Bruxelles (ULB)\\
  Brussels, Belgium\\
  \texttt{\{jcardin,aureooms\}@ulb.ac.be}}
\affil[2]{New York University\\ New York, United States of America\\
  \texttt{icalp2016submission@johniacono.com}}
\authorrunning{J.\,Cardinal, J.\,Iacono, and A.\,Ooms} 
\Copyright{Jean Cardinal, John Iacono, and Aur\'elien Ooms}

\serieslogo{}
\volumeinfo
  {Billy Editor and Bill Editors}
  {2}
  {Conference title on which this volume is based on}
  {1}
  {1}
  {1}
\EventShortName{}
\DOI{10.4230/LIPIcs.xxx.yyy.p}

\begin{document}

\maketitle


\begin{abstract}
The \kSUM{} problem is given $n$ input real numbers to determine whether any
$k$ of them sum to zero. The problem is of tremendous importance in the
emerging field of complexity theory within $P$, and it is in particular open
whether it admits an algorithm of complexity $O(n^c)$ with $c<\lceil
\frac{k}{2} \rceil$. Inspired by an algorithm due to Meiser (1993), we show
that there exist linear decision trees and algebraic computation trees of depth
$O(n^3\log^3 n)$ solving \kSUM{}. Furthermore, we show that there exists a
randomized algorithm that runs in \(\tilde{O}(n^{\ceil{\frac{k}{2}}+8})\) time,
and performs $O(n^3\log^3 n)$ linear queries on the input. Thus, we show that
it is possible to have an algorithm with a runtime almost identical (up to the
$+8$) to the best known algorithm but for the first time also with the number
of queries on the input a polynomial that is independent of $k$. The
$O(n^3\log^3 n)$ bound on the number of linear queries is also a tighter bound
than any known algorithm solving \kSUM{}, even allowing unlimited total time
outside of the queries.
By simultaneously achieving few queries to the input without significantly
sacrificing runtime vis-\`{a}-vis known algorithms, we deepen the understanding
of this canonical problem which is a cornerstone of complexity-within-$P$.

We also consider a range of tradeoffs between the number of terms involved in
the queries and the depth of the decision tree. In particular, we prove that
there exist $o(n)$-linear decision trees of depth $o(n^4)$.

\end{abstract}


\section{Introduction}

The \kSUM{} problem is defined as follows: given a collection of real numbers, and a constant $k$,
decide whether any $k$ of them sum to zero. 
It is a fixed-parameter version of the subset-sum problem, a standard \textit{NP}-complete
problem. The \kSUM{} problem, and in particular the special case of \threeSUM{},
has proved to be a cornerstone of the fine-grained complexity program aiming at
the construction of a complexity theory for problems in $P$. In particular,
there are deep connections between the complexity of \kSUM{}, the Strong
Exponential Time Hypothesis~\cite{PW10,CGIMPS15}, and the complexity of many
other major problems in
$P$~\cite{GO95,BH99,MO01,P10,ACLL14,AVW14,GP14,KPP14,ALW15,AWY15,CL15}.

It has been long known that the \kSUM{} problem can be solved in time
$O(n^{\frac{k}{2}}\log n)$ for even $k$, and $O(n^{\frac{k+1}{2}})$ for odd
$k$. Erickson~\cite{E99} proved a near-matching lower bound in the $k$-linear
decision tree model. In this model, the complexity is measured by the depth of
a decision tree, every node of which corresponds to a query of the form
$q_{i_1} + q_{i_2} + \cdots + q_{i_k} \ask{\le} 0$, where $q_1, q_2, \ldots, q_n$ are the
input numbers. In a recent breakthrough paper, Gr\o nlund and
Pettie~\cite{GP14} showed that in the $(2k-2)$-linear decision tree model,
where queries test the sign of weighted sums of up to $2k-2$ input numbers, only
$O(n^\frac{k}{2}\sqrt{\log n})$ queries are required for odd values of $k$. In
particular, there exists a $4$-linear decision tree for \threeSUM{} of depth
$\tilde{O}(n^\frac{3}{2})$ (here the notation $\tilde{O}$ ignores polylogarithmic
factors), while every 3-linear decision tree has depth $\Omega
(n^2)$~\cite{E99}. This indicates that increasing the size of the queries,
defined as the maximum number of input numbers involved in a query, can yield
significant improvements on the depth of the minimal-height decision tree. Ailon and
Chazelle~\cite{AC05} slightly extended the range of query sizes for which a
nontrivial lower bound could be established, elaborating on Erickson's
technique.

It has been well established that there exist nonuniform
polynomial-time algorithms for the subset-sum problem. One of them was
described by Meiser~\cite{M93}, and is derived from a data structure for point
location in arrangements of hyperplanes using the bottom vertex decomposition.
This algorithm can be cast as the construction of a linear decision tree in which
the queries have non-constant size.

\subsection{Our results.}
Our first contribution, in Section~\ref{sec:query-complexity}, is through a
careful implementation of Meiser's basic algorithm idea~\cite{M93} to show the
existence of an $n$-linear decision tree of depth $\tilde{O}(n^3)$ for \kSUM{}.
Although the high-level algorithm itself is not new, we refine the
implementation and analysis for the \kSUM{} problem.
Meiser presented his algorithm as a general method of point location in $\H{}$
given $n$-dimensional hyperplanes that yielded a depth $\tilde{O}(n^4 \log
\card{\H{}})$ linear decision tree; when viewing the \kSUM{} problem as a point
location problem, $\H{}$ is $O(n^k)$ and thus Meiser's algorithm can be viewed
as giving a $\tilde{O}(n^4)$-height linear decision tree. We reduce this height
to $\tilde{O}(n^3)$ for the particular hyperplanes that an instance $\kSUM$
problem presents when viewed as a point location problem.
While the original algorithm was cast as a nonuniform polynomial-time
algorithm, we show that it can be implemented in the linear decision tree
model. This first result also implies the existence of nonuniform RAM
algorithms with the same complexity, as shown in Appendix~\ref{app:act}.

There are two subtleties to this result. The first is inherent to the chosen
complexity model: even if the number of queries to the input is small (in
particular, the degree of the polynomial complexity is invariant on $k$), the
time required to \emph{determine which queries should be performed} may be
arbitrary. In a na\"ive analysis, we show it can be trivially bounded by
$\tilde{O}(n^{k+2})$. In Section~\ref{sec:time-complexity} we present an algorithm to
choose
which decisions to perform whereby the running time can be reduced to
$\tilde{O}(n^{\frac{k}{2}+8})$. Hence, we obtain an
$\tilde{O}(n^{\frac{k}{2}+8})$ time randomized algorithm in the RAM model
expected to perform $\tilde{O}(n^3)$ linear queries on the input\footnote{%
	Gr{\o}nlund and
	Pettie~\cite{GP14} mention the algorithms of auf Der Heyde~\cite{A84} and
	Meiser~\cite{M93}, and state ``(\ldots) it was known that all \kLDT{} problems
	can be solved by $n$-linear decision trees with depth $O(n^5\log
	n)$~\cite{M93}, or with depth $O(n^4\log (nK))$ if the coefficients of the
	linear function are integers with absolute value at most $K$~\cite{A84}.
	Unfortunately these decision trees are not efficiently constructible. The
	time required to determine \emph{which} comparisons to make is
	exponential.'' We prove that the trees can have depth $\tilde{O}(n^3)$ and
	that the whole algorithm can run in randomized polynomial-time.}.

The second issue we address is that the linear queries in the above algorithm may
have size $n$, that is, they may use all the components of the input.
The lower bound of Erickson shows that if the queries are of minimal size, the number
of queries cannot be a polynomial independent of $k$ such as what we obtain, so
non-minimal query size is clearly essential to a drastic reduction in the number of queries needed.
This gives rise to the natural question as to what is the relation between query size and number of queries.
In particular, one natural question is whether queries of size less than $n$
would still allow the problem to be solved using a polynomial independent of $k$
number of queries. We show that this is possible; in Section~\ref{sec:query-size},
we introduce a range of
algorithms exhibiting an explicit tradeoff between the number of queries and
their size. Using a blocking scheme, we show that we can restrict to
$\varepsilon n$-linear decision trees, for arbitrary $\varepsilon$. We also give a
range of tradeoffs for $n^{1-\varepsilon}$-linear decision trees. Although the
proposed algorithms still involve nonconstant-size queries, this is
the first time such tradeoffs are explicitly tackled. Table~\ref{tab:results} summarizes our results.

\begin{table}
\centering
\begin{tabular}{|c|c|c|c|c|}
	\hline

	& \# blocks & query size & \# queries & time \\
	\hline
	Theorem~\ref{thm:cube} &
	$1$ &
	$n$ &
	$\tilde{O}(n^3)$ & $\tilde{O}(n^{\ceil{\frac{k}{2}}+8})$
	\\

	\hline

	Theorem~\ref{thm:query-size} &
	$b$ &
	$k\ceil{\frac{n}{b}}$ &
	$\tilde{O}(b^{k-4}n^3)$ &
	$\tilde{O}(b^{\floor{\frac{k}{2}}-9} n^{\ceil{\frac{k}{2}}+8})$
	\\

	\hline

	Corollary~\ref{cor:en} &
	$b=\frac{k}{\varepsilon}$ &
	$\varepsilon n$ &
	$\tilde{O}(n^3)$ &
	$\tilde{O}(n^{\ceil{\frac{k}{2}}+8})$
	\\

	\hline

	Corollary~\ref{cor:ne} &
	$b=O(n^\varepsilon)$ &
	$O(n^{1-\varepsilon})$ &
	$\tilde{O}(n^{3+(k-4)\varepsilon})$ &
	$\tilde{O}(n^{(1+\varepsilon)\frac{k}{2} +8.5})$
	\\

	\hline
\end{tabular}
\caption{Complexities of our new algorithms for the \kSUM{}\ problem. The query
size is the maximum number of elements of the input that can be involved in a
single linear query. The number of blocks is a parameter that allows us to
change the query size (see Section~\ref{sec:query-size}).
The origin of the constant in the exponent of the time complexity is due to
Lemma~\ref{lem:bound}. We conjecture it can be reduced, though substantial
changes in the analysis will likely be needed to do so.}
\label{tab:results}
\end{table}


\section{Definitions and previous work}

\subsection{Definitions}
We consider the \kSUM{} problem for \(k=O(1)\). In what follows, we use the
notation \([n] = \{\,1,2,\ldots ,n\,\}\).
\begin{problem}[\kSUM{}]
 Given an input vector \(q\in\mathbb{R}^n\), decide whether there exists a
 $k$-tuple \((i_1, i_2,\ldots ,i_k) \in {[n]}^k\) such that \(\sum_{j=1}^k
 q_{i_j} = 0\).
\end{problem}
The problem amounts to deciding in $n$-dimensional space, for each hyperplane
\(H\) of equation \(x_{i_1} + x_{i_2} + \cdots +x_{i_k} = 0\), whether \(q\)
lies on, above, or below \(H\). Hence this indeed amounts to locating the point
$q$ in the arrangement formed by those hyperplanes. We emphasize that the set
of hyperplanes depends only on $k$ and $n$ and not on the actual input vector
$q$.

Linear degeneracy testing (\kLDT{}) is a generalization of \kSUM{} where we
have arbitrary rational coefficients\footnote{The usual definition of \kLDT{}
allows arbitrary \emph{real} coefficients. However, the algorithm we provide
for Lemma~\ref{lem:multiple} needs the vertices of the arrangement of
hyperplanes to have rational coordinates.}
and an independent term in the equations
of the hyperplanes.
\begin{problem}[\kLDT{}]
 Given an input vectors \(q\in\mathbb{R}^n\) and
 $\alpha \in \mathbb{Q}^n$ and constant $c \in \mathbb{Q}$
 decide whether there exists a
 $k$-tuple \((i_1, i_2,\ldots ,i_k) \in {[n]}^k\) such that
 \(c + \sum_{j=1}^k \alpha_j q_{i_j} = 0\).
 \end{problem}
Our algorithms apply to this more general problem with only minor changes.

The \emph{\(s\)-linear decision tree model} is a standard model of computation
in which several lower bounds for \kSUM{}\ have been proven. In the decision tree
model, one may ask well-defined questions to an oracle that are answered
``yes'' or ``no.'' For $s$-linear decision trees, a well-defined question consists
of testing the sign of a linear function on at most \(s\) numbers \(q_{i_1},\ldots,q_{i_s}\) of the
input \(q_1,\ldots,q_n\) and can be written as
$$
	c + \alpha_1 q_{i_1} + \cdots + \alpha_s q_{i_s} \ask{\le} 0
$$
Each question is defined to cost a single unit. All other operations can be
carried out for free but may not examine the input vector $q$. We refer to
$n$-linear decision trees simply as linear decision trees.

In this paper, we consider algorithms in the standard integer RAM model, but in which
the input $q\in\mathbb{R}^n$ is accessible \emph{only} via a linear query oracle. Hence we
are not allowed to manipulate the input numbers directly. The complexity is measured
in two ways: by counting the total number of queries, just as in the linear decision tree
model, and by measuring the overall running time, taking into account the time required
to determine the sequence of linear queries. This two-track computation model,
in which the running time is distinguished from the query complexity, is commonly
used in results on comparison-based sorting problems where analyses of both
runtime and comparisons are of interest (see for
instance~\cite{SS95,CFJJM10,CFJJM13}).

\subsection{Previous Results.}
The seminal paper by Gajentaan and Overmars~\cite{GO95} showed the crucial role
of \threeSUM{} in understanding the complexity of several problems in
computational geometry. It is an open question whether an
$O(n^{2-\varepsilon})$ algorithm exists for \threeSUM{}. Such a result would
have a tremendous impact on many other fundamental computational
problems~\cite{GO95,BH99,MO01,P10,ACLL14,AVW14,GP14,KPP14,ALW15,AWY15,CL15}.
It has been known for long that \kSUM{} is $W[1]$-hard. Recently, it was shown
to be $W[1]$-complete by Abboud et al.~\cite{ALW15}.

In Erickson~\cite{E99}, it is shown that we cannot solve \threeSUM{} in
subquadratic time in the \(3\)-linear decision tree model:
\begin{theorem}[Erickson~\cite{E99}]
The optimal depth of a \(k\)-linear decision tree that solves
the \kLDT{} problem is $\Theta(n^{\ceil{\frac{k}{2}}})$.
\end{theorem}
The proof uses an adversary argument which can be explained geometrically. As
we already observed, we can solve \kLDT{} problems by modeling them as point
location problems in an arrangement of hyperplanes. Solving one such problem
amounts to determining which cell of the arrangement contains the input point.
The adversary argument of Erickson~\cite{E99} is that there exists a cell having
$\Omega(n^{\ceil{\frac{k}{2}}})$ boundary facets and in this model point
location in such a cell requires testing each facet.

Ailon and Chazelle~\cite{AC05} study \(s\)-linear decision trees to solve the \kSUM{} problem when
\(s > k\). In particular, they give an additional proof for the
$\Omega(n^{\ceil{\frac{k}{2}}})$ lower bound of Erickson~\cite{E99} and
generalize the lower bound for the \(s\)-linear decision tree model when \(s >
k\). Note that the exact lower bound given by Erickson~\cite{E99} for \(s = k\) is
$\Omega({(nk^{-k})}^{\ceil{\frac{k}{2}}})$ while the one given by
Ailon and Chazelle~\cite{AC05} is $\Omega({(nk^{-3})}^{\ceil{\frac{k}{2}}})$. Their result
improves therefore the lower bound for \(s = k\) when \(k\) is large.
The lower bound they prove for \(s > k\) is the following
\begin{theorem}[Ailon and Chazelle~\cite{AC05}]
The depth of an $s$-linear decision tree solving the \kLDT{}\ problem is
$$
\Omega\mleft(\group{nk^{-3}}^{\frac{2k-s}{2\ceil{\frac{s-k+1}{2}}}
(1-\varepsilon_k) }\mright),
$$
where \(\varepsilon_k > 0\) tends to \(0\) as \(k \to\infty\).
\end{theorem}
This lower bound breaks down when
\(k = \Omega(n^{\sfrac{1}{3}})\) or \(s \ge 2 k\) and the cases where \(k < 6\)
give trivial lower bounds. For example, in the case
of \threeSUM{} with \(s = 4\) we only get an $\Omega(n)$ lower bound.

As for upper bounds, Baran et al.~\cite{BDP08} gave subquadratic Las Vegas
algorithms for 3SUM on integer and
rational numbers in the circuit RAM, word RAM, external memory, and
cache-oblivious models of computation. The idea of their approach is to exploit
the parallelism of the models, using linear and universal hashing.


More recently, Gr{\o}nlund and Pettie~\cite{GP14} proved the existence of a linear decision tree
solving the \threeSUM{} problem using a strongly subquadratic number of linear queries.
The classical quadratic algorithm for \threeSUM{} uses \(3\)-linear queries
while the decision tree of Gr{\o}nlund and Pettie uses \(4\)-linear queries and
requires $O(n^{\sfrac{3}{2}} \sqrt{\log n})$ of them.
Moreover, they show that their decision tree can be used to get better upper
bounds for \kSUM{} when \(k\) is odd.

They also provide two subquadratic \threeSUM{}
algorithms. A deterministic one running in
$O(n^2/{(\log n/\log \log n)}^{\sfrac{2}{3}})$
time and a randomized one running in
$O(n^2 {(\log \log n)}^2 / \log n)$ time with high probability.
These results refuted the long-lived conjecture that
\threeSUM{} cannot be solved in subquadratic time in the RAM model.

Freund~\cite{F15} and Gold and Sharir~\cite{GS15} later gave improvements on the
results of Gr{\o}nlund and Pettie~\cite{GP14}. Freund~\cite{F15} gave a deterministic algorithm for
\threeSUM{} running in \(O( {n^2\log \log n}/{\log n})\) time.
Gold and Sharir~\cite{GS15} gave another deterministic algorithm for \threeSUM{}
with the same running time and shaved off the $\sqrt{\log n}$ factor in the
decision tree complexities of \threeSUM{} and \kSUM{} given by Gr{\o}nlund and Pettie.

Auf der Heide~\cite{A84} gave the first point location algorithm to solve the knapsack
problem in the linear decision tree model in polynomial time. He thereby
answers a question raised by Dobkin and Lipton~\cite{DL74,DL78}, Yao~\cite{Y82}
and others. However, if one uses this algorithm to locate a point in an
arbitrary arrangement of hyperplanes the running time is increased by a factor
linear in the greatest coefficient in the equations of all hyperplanes.
On the other hand, the complexity of Meiser's point location algorithm is
polynomial in the dimension, logarithmic in the number of hyperplanes and
does not depend on the value of the coefficients in the equations of the
hyperplanes. A useful complete description of the latter is also given by
Bürgisser et al.~\cite{B97} (Section 3.4).


\section{Query complexity}
\label{sec:query-complexity}

In this section and the next, we prove the following first result.
\begin{theorem}
\label{thm:cube}
There exist linear decision trees of depth at most \(O(n^3\log^3 n)\) solving
the \kSUM{} and the \kLDT{} problems. Furthermore, for the two problems there
exists an $\tilde{O}(n^{\ceil{\frac{k}{2}}+8})$ Las Vegas algorithm in
the RAM model expected to perform $O(n^3\log^3 n)$ linear queries on the input.
\end{theorem}

\subsection{Algorithm outline}
For a fixed set of hyperplanes \(\H\) and given input vertex \(q\) in \(\R^n\),
Meiser's algorithm allows us to determine the cell of the arrangement
$\arrangement(\H)$ that contains $q$ in its interior (or that \emph{is} $q$ if
$q$ is a $0$-cell of $\arrangement(\H)$), that is, the positions $\sigma(H,q) \in
\signset$ of \(q\) with respect to all hyperplanes $H \in \H$. In the \kSUM{}
problem, the set $\H$ is the set of $\Theta(n^k)$ hyperplanes with equations of the
form $x_{i_1} + x_{i_2} + \cdots + x_{i_k} = 0$. 
These equations can be modified accordingly for \kLDT{}.

We use standard results on \enets{}. Using a theorem due to Haussler and
Welzl~\cite{H87}, it is possible to construct an \enet{} \(\NH\) for the range
space
defined by hyperplanes and simplices using a random uniform sampling on \(\H\).
\begin{theorem}[Haussler and Welzl~\cite{H87}]\label{thm:enet}
	If we choose \(O(\enetsize)\) hyperplanes of \(\H\) uniformly at
random and denote this selection \(\net\) then for
any simplex intersected by more than \(\varepsilon \card{\H{}}\) hyperplanes of
\(\H\), with high probability, at least one of the intersecting hyperplanes
is contained in \(\net\).
\end{theorem}
The contrapositive states that if no hyperplane in \(\net\) intersects
a given simplex, then with high probability the number of hyperplanes of
\(\H\) intersecting the simplex is at most \(\varepsilon \card{\H{}}\).

We can use this to design a prune and search algorithm as follows:
(1) construct an \enet{} \(\net\),
(2) compute the cell \(\cell\) of \(\arrangement(\net)\) containing the input
point $q$ in its interior,
(3) construct a simplex \(\simplex\) inscribed in \(\cell\) and containing
\(q\) in its interior,
(4) recurse on the hyperplanes of \(\H\) intersecting the interior of
\(\simplex\).

Proceeding this way with a constant $\varepsilon$ guarantees that at most a
constant fraction \(\varepsilon\) of the hyperplanes remains after the pruning step,
and thus the cumulative number of queries made to determine the enclosing cell at
each step is $O(n^2 \log^2 n \log \card{\H{}})$ when done in a brute-force way.
However, we still need to explain how to find a simplex \(\simplex\) inscribed
in \(\cell\) and containing \(q\) in its interior. This procedure corresponds to the well-known
\emph{bottom vertex decomposition} (or \emph{triangulation}) of a hyperplane
arrangement~\cite{GO04,Cla88}.

\subsection{Finding a simplex}

In order to simplify the exposition of the algorithm, we assume, without
loss of generality, that the input numbers $q_i$ all lie in the interval $[-1,1]$.
This assumption is justified by observing that we can normalize all the input
numbers by the largest absolute value of a component of $q$. One can then see that
every linear query on the normalized input can be implemented
as a linear query on the original input. A similar transformation can be carried out
for the \kLDT{}\ problem.
This allows us to use bounding hyperplanes of equations $x_i = \pm 1, i\in [n]$.
We denote by $\B$ this set of hyperplanes. Hence, if we choose a subset
\(\net\) of the hyperplanes, the input point is located in a bounded cell
of the arrangement \(\arrangement(\net \cup \B)\). Note that \(\card{\net \cup
\B} = O(\card{\net})\) for all interesting values of \(\varepsilon\).


We now explain how to construct \(\simplex\) under this assumption. The algorithm
can be sketched as follows. (Recall that $\sigma(H,p)$ denotes the relative position
of $p$ with respect to the hyperplane $H$.)

\begin{algorithm}[Constructing \(\simplex\)]\label{algo:simplex}
\item[input] A point \(q\) in ${[-1,1]}^n$, a set $\I$ of hyperplanes not
	containing \(q\), and a set $\E$ of hyperplanes in general position
	containing \(q\), such that the cell
	$$
	\cell = \{\,p \sut \sigma(H,p) = \sigma(H,q)
			\ \text{or}\ \sigma(H,p) = 0
			\ \text{for all}\ H \in (\I \cup \E)
		\,\}
	$$
	is a bounded polytope. The value \(\sigma(H,q)\) is known for
	all \(H \in (\I \cup \E)\). 
\item[output] A simplex \(\simplex \in \cell\) that contains \(q\) in
	its interior (if it is not a point), and all vertices
	of which are vertices of \(\cell\).
\item[0.] If $\card{\E}=n$, return $q$.
\item[1.] Determine a vertex \(\nu\) of $\cell$.
\item[2.] Let \(q'\) be the projection of \(q\) along \(\vec{\nu q}\) on the
	boundary of \(\cell\). Compute \(\I_\theta \subseteq \I\), the subset of
	hyperplanes in \(\I\) containing \(q'\). Compute \(\I_{\tau} \subseteq
	\I_{\theta}\), a maximal subset of those hyperplanes such that \(\E' = \E
	\cup \I_\tau\) is a set of hyperplanes in general position.
\item[3.] Recurse on \(q'\), \(\I' = \I \setminus \I_\theta\), and \(\E'\), and
	store the result in \(\simplex'\).
\item[4.] Return $\simplex$, the convex hull of \(\simplex' \cup \enum{\nu}\).
\end{algorithm}

Step \step{0} is the base case of the recursion: when there is only one point left, just return
that point. This step uses no query.

We can solve step \step{1} by using linear programming with the known values
of \(\sigma(H,q)\) as linear constraints. We arbitrarily choose an
objective function with a gradient non-orthogonal to all hyperplanes in
\(\I\) and look for the optimal solution. The optimal solution being a vertex of the arrangement,
its coordinates are independent of \(q\), and thus this step involves no query at all.

Step \step{2} prepares the recursive step by finding the hyperplanes containing
\(q'\). This can be implemented as a ray-shooting algorithm that performs
a number of comparisons between projections of $q$ on different hyperplanes of $\I$ without
explicitly computing them. In~Appendix~\ref{app:keeplinear}, we prove that all such comparisons
can be implemented using \(O(\card{\I})\) linear queries.
Constructing \(\E'\) can be done by solving systems of linear
equations that do not involve \(q\).

In step \step{3}, the input conditions are satisfied, that is, $q' \in
{[-1,1]}^n$, \(\I'\) is a set of hyperplanes not containing \(q'\), \(\E'\) is
a set of hyperplanes in general position containing \(q'\), \(\cell'\) is a
$d$-cell of \(\cell\) and is thus a bounded polytope. The value \(\sigma(H,q')\)
differs from \(\sigma(H,q)\) only for hyperplanes that have been removed from
\(\I\), and for those \(\sigma(H,q') = 0\), hence we know all necessary values
\(\sigma(H,q')\) in advance.

Since \(\card{\I'} < \card{\I}\), \(\card{\E'} > \card{\E}\), and
\(\card{\I\setminus\I'} - \card{\E'\setminus\E} \ge 0\), the complexity of the
recursive call is no more than that of the parent call, and the maximal depth
of the recursion is \(n\). Thus, the total number of
linear queries made to compute \(\simplex\) is \(O(n\card{\I})\).

Hence given an input point \(q \in [-1,1]\), an arrangement of hyperplanes
\(\arrangement(\net)\), and the value of \(\sigma(H,q)\) for all
\(H \in (\net\cup\B)\), we can compute the desired simplex \(\simplex\)
by running Algorithm~\ref{algo:simplex} on \(q\),
\(\I=\{\,H \in (\net\cup\B) \sut \sigma(H,q) \neq 0\,\}\), and
\(\E \subseteq (\net\cup\B) \setminus\I\).
This uses $O(n^3 \log^2 n)$ linear queries.
Figure~\ref{fig:meiser:step} illustrates a step of the algorithm.

\begin{figure}
\begin{center}
\includegraphics[trim=90 47 50 13,clip=true,height=0.25\textheight]{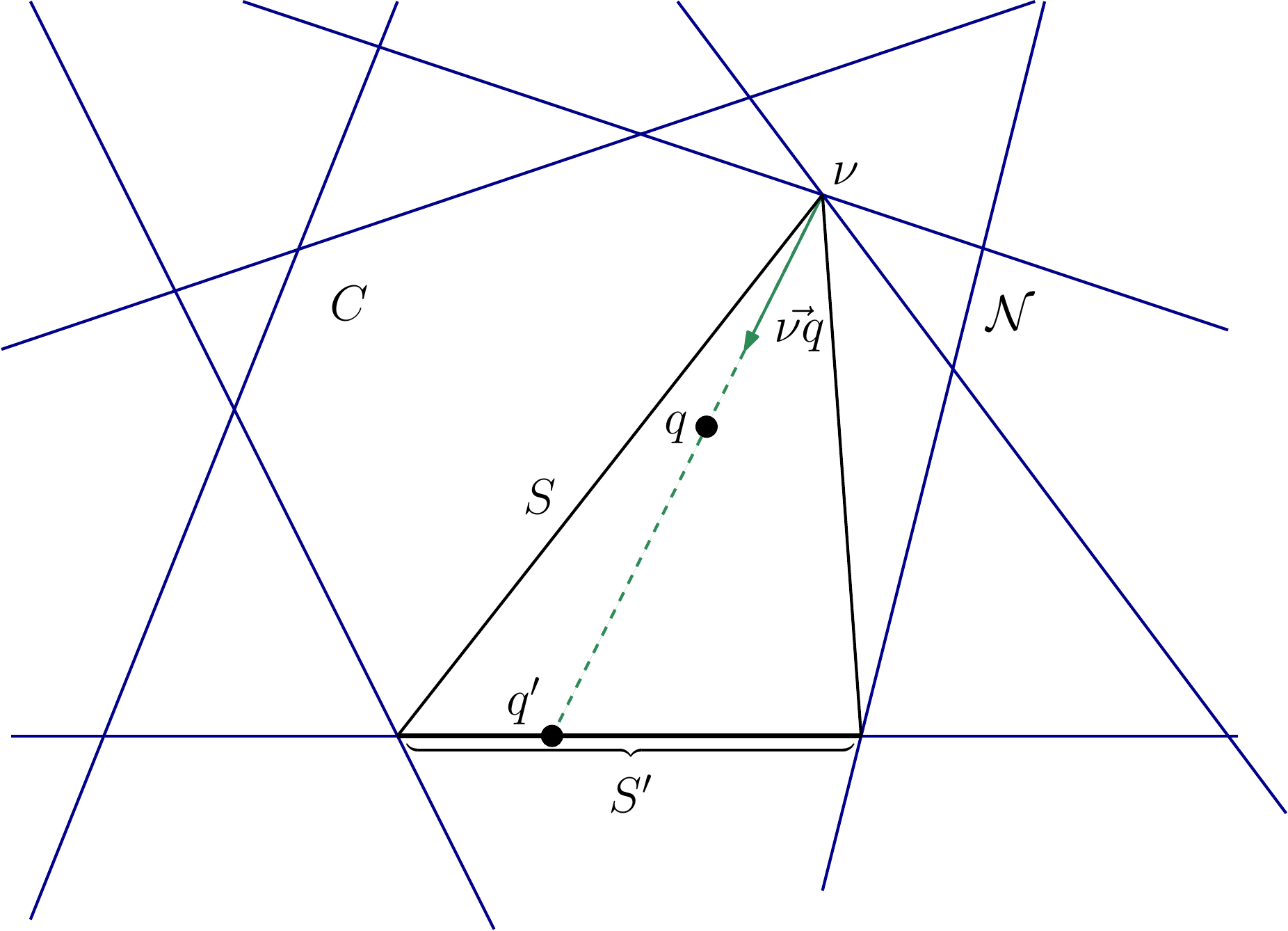}
\caption{%
Illustration of a step of Algorithm~\ref{algo:simplex}.
}
\label{fig:meiser:step}
\end{center}
\end{figure}

\subsection{Assembling the pieces}

Let us summarize the algorithm
\begin{algorithm}\label{algo:meiser}
\item[input] \(q \in {[-1,1]}^n\)
\item[1.] Pick \(O(n^2 \log^2 n)\) hyperplanes of $\H$ at random and locate $q$
in this arrangement. Call $\cell$ the cell containing $q$.
\item[2.] Construct the simplex \(\simplex\) containing \(q\) and inscribed in
\(\cell\), using Algorithm~\ref{algo:simplex}.
\item[3.] For every hyperplane of $\H$ containing $\simplex$, output a solution.
\item[4.] Recurse on hyperplanes of $\H$ intersecting the interior of $\simplex$.
\end{algorithm}

The query complexity of step \step{1} is $O(n^2 \log^2 n)$, and that of step
\step{2} is $O(n^3 \log^2 n)$. Steps \step{3} and \step{4} do not involve any
query at all. The recursion depth is $O(\log \card{\H})$, with $|\H|=O(n^k)$,
hence the total query complexity of this algorithm is $O(n^3 \log^3 n)$. This
proves the first part of Theorem~\ref{thm:cube}.\\

We can also consider the overall complexity of the algorithm in the RAM model,
that is, taking into account the steps that do not require any query, but for
which we still have to process the set $\H$. Note that the complexity
bottleneck of the algorithm are steps \step{3}-\step{4}, where we need to prune
the list of hyperplanes according to their relative positions with respect to
$\simplex$. For this purpose, we simply maintain explicitly the list of all
hyperplanes, starting with the initial set corresponding to all $k$-tuples.
Then the pruning step can be performed by looking at the position of each
vertex of $\simplex$ relative to each hyperplane of $\H$.
Because in our case hyperplanes have only $k$ nonzero coefficients, this uses a
number of integer arithmetic operations on $\tilde{O}(n)$ bits integers that is
proportional to the number of vertices times the number of hyperplanes.
(For the justification of the bound on the number of bits needed to represent
vertices of the arrangement see Appendix~\ref{app:bound}.)
Since we recurse on a fraction of the set, the overall complexity is
$\tilde{O}(n^2\card{\H{}}) = \tilde{O}(n^{k+2})$. The next section is devoted
to improving this running time.


\section{Time complexity}
\label{sec:time-complexity}

Proving the second part of Theorem~\ref{thm:cube} involves efficient implementations of
the two most time-consuming steps of Algorithm~\ref{algo:meiser}.
In order to efficiently implement the pruning
step, we define an intermediate problem, that we call the \emph{double $k$-SUM}
problem.
\begin{problem}[double $k$-SUM]
	Given two vectors $\nu_1, \nu_2 \in {[-1,1]}^n$, where the coordinates of
	$\nu_i$ can be written down as fractions whose numerator and denominator
	lie in the interval $[-M,M]$, enumerate all
	$i\in {[n]}^k$ such that
	$$
	\left(\sum_{j=1}^{k} \nu_{1,i_j}\right)
	\left(\sum_{j=1}^{k} \nu_{2,i_j}\right)
	< 0.
	$$
\end{problem}

In other words, we wish to list all hyperplanes of $\H$
intersecting the open line segment $\nu_1\nu_2$.
We give an efficient output-sensitive algorithm for this problem.

\begin{lemma}\label{lem:double}
	The double $k$-SUM problem can be solved in time $O(n^{\lceil \frac{k}{2}
	\rceil} \log n \log M + Z)$, where $Z$ is the size of the solution.
\end{lemma}
\begin{proof}
	If $k$ is even, we consider all possible $\frac{k}{2}$-tuples of numbers in
	$\nu_1$ and $\nu_2$ and sort their sums in increasing order. This takes time
	$O(n^{\frac{k}{2}} \log n)$ and yields two permutations $\pi_1$ and $\pi_2$
	of $[n^\frac{k}{2}]$.
	If $k$ is odd, then we sort both the $\lceil\frac{k}{2}\rceil$-tuples and
	the $\lfloor\frac{k}{2}\rfloor$-tuples. For simplicity, we will only
	consider the even case in what follows. The odd case carries through.

	We let $N = n^\frac{k}{2}$. For $i \in [N]$ and $m \in \{1,2\}$, let
	$\Sigma_{m,i}$ be the sum of the $\frac{k}{2}$ components of the $i$th
	$\frac{k}{2}$-tuple in $\nu_m$, in the order prescribed by $\pi_m$.

	We now consider the two $N \times N$ matrices $M_1$ and $M_2$ giving all
	possible sums of two $\frac{k}{2}$-tuples, for both $\nu_1$ with the ordering
	$\pi_1$ and $\nu_2$ with the ordering $\pi_2$.

	We first solve the $k$-SUM problem on $\nu_1$, by finding the sign of all
	pairs $\Sigma_{1,i} + \Sigma_{1,j}$, $i, j \in [N]$. This can be done in
	time $O(N)$ by parsing the matrix $M_1$, just as in the standard \kSUM{} algorithm.
        We do the same with $M_2$.

	The set of all indices $i, j \in [N]$ such that $\Sigma_{1,i} +
	\Sigma_{1,j}$ is positive forms a staircase in $M_1$. We sweep $M_1$ column
	by column in order of increasing $j \in [N]$, in such a way that the number
	of indices $i$ such that $\Sigma_{1,i} + \Sigma_{1,j} > 0$ is growing.
	For each new such value $i$ that is encountered during the sweep, we insert
	the corresponding $i' = \pi_2(\pi_1^{-1}(i))$ in a balanced binary search
	tree.

	After each sweep step in $M_1$ --- that is, after incrementing $j$ and
	adding the set of new indices $i'$ in the tree --- we search the tree to
	identify all the indices $i'$ such that $\Sigma_{2,i'} + \Sigma_{2,j'} <
	0$, where $j' = \pi_2(\pi_1^{-1}(j))$. Since those indices form an interval
	in the ordering $\pi_2$ when restricted to the indices in the tree, we can
	search for the largest $i_0'$ such that $\Sigma_{2,i_0'} < -\Sigma_{2,j'}$
	and retain all indices $i' \le i_0'$ that are in the tree. If we denote by
	$z$ the number of such indices, this can be done in $O(\log N + z) = O(\log
	n + z)$ time. Now all the pairs $i', j'$ found in this way are such that
	$\Sigma_{1,i} + \Sigma_{1,j}$ is positive and $\Sigma_{2,i'} +
	\Sigma_{2,j'}$ is negative, hence we can output the corresponding
	$k$-tuples. To get all the pairs $i', j'$ such that
	$\Sigma_{1,i} + \Sigma_{1,j}$ is negative and $\Sigma_{2,i'} +
	\Sigma_{2,j'}$ positive, we repeat the sweeping algorithm after swapping the
	roles of $\nu_1$ and $\nu_2$.

	Every matching $k$-tuple is output exactly once, and every
	$\frac{k}{2}$-tuple is inserted at most once in the binary search tree.
	Hence the algorithm runs in the claimed time.

	Note that we only manipulate rational numbers that are the sum of at most
	$k$ rational numbers of size $O(\log M)$.
\end{proof}

Now observe that a hyperplane intersects the interior of a simplex if and only
if it intersects the interior of one of its edges. Hence given a simplex
$\simplex$ we can find all hyperplanes of $\H$ intersecting its interior by
running the above algorithm ${n\choose 2}$ times, once for each pair of
vertices $(\nu_1,\nu_2)$ of $\simplex$, and take the union of the solutions.
The overall running time for this implementation will therefore be
$\tilde{O}(n^2 (n^{\lceil \frac k2 \rceil} \log M + Z))$, where $Z$ is at most the
number of intersecting hyperplanes and $M$ is to be determined later.
This provides an implementation of the pruning step in Meiser's algorithm, that
is, step~\step{4} of Algorithm~\ref{algo:meiser}.

\begin{corollary}\label{cor:double}
	Given a simplex $\simplex$, we can compute all $k$-SUM hyperplanes
	intersecting its interior in $\tilde{O}(n^2 (n^{\lceil \frac k2 \rceil}
	\log M + Z))$ time, where $M$ is proportional to the number of bits
	necessary to represent $\simplex$.
\end{corollary}

In order to detect solutions in step~\step{3} of Algorithm~\ref{algo:meiser}, we also
need to be able to quickly solve the following problem.
\begin{problem}[multiple $k$-SUM]
	Given $d$ points $\nu_1,\nu_2,\ldots,\nu_d \in \R^n$, where the coordinates of
	$\nu_i$ can be written down as fractions whose numerator and denominator
	lie in the interval $[-M,M]$,
	decide whether there exists a hyperplane with equation of the form $x_{i_1}
	+ x_{i_2} + \cdots + x_{i_k} = 0$ containing all of them.
\end{problem}

Here the standard $k$-SUM algorithm can be applied, taking advantage of
the fact that the coordinates lie in a small discrete set.
\begin{lemma}\label{lem:ksum}
	$k$-SUM on $n$ integers $\in [-V,V]$ can be solved in time
	$\tilde{O}(n^{\ckt}\log V)$.
\end{lemma}

\begin{lemma}\label{lem:multiple}
	Multiple $k$-SUM can be solved in time $\tilde{O}(dn^{\ckt+2}\log M)$.
\end{lemma}
\begin{proof}
	Let $\mu_{i,j}$ and $\delta_{i,j}$ be the numerator and denominator of
	$\nu_{i,j}$ when written as an irreducible fraction. We define
	$$
		\zeta_{i,j} = \nu_{i,j}\prod_{\mathclap{(i,j) \in [d]\times[n]}} \delta_{i,j} =
		\frac{\mu_{i,j}\displaystyle\prod_{\mathclap{(i',j') \in [d]\times[n]}} \delta_{i',j'}}{\delta_{i,j}}.
	$$
	By definition $\zeta_{i,j}$ is an integer and its absolute value is bounded
	by $U = M^{n^2}$, that is, it can be
	represented using $O(n^2 \log M)$ bits. Moreover, if one of the hyperplanes
	contains the point $(\zeta_{i,1},\zeta_{i,2},\ldots,\zeta_{i,n})$, then it
	contains $\nu_i$. Construct $n$ integers of $O(dn^2 \log M)$ bits that can
	be written $\zeta_{1,j}+U,\zeta_{2,j}+U,\ldots,\zeta_{d,j}+U$ in base
	$2Uk+1$. The answer to our decision problem is ``yes'' if and only if there
	exists $k$ of those numbers whose sum is $kU,kU,\ldots,kU$. We simply
	subtract the number $U,U,\ldots,U$ to all $n$
	input numbers to obtain a standard $k$-SUM instance on $n$ integers of
	$O(dn^2 \log M)$ bits.
\end{proof}

We now have efficient implementations of steps \step{3} and \step{4}
of Algorithm~\ref{algo:meiser} and can proceed to the proof of the second part
of Theorem~\ref{thm:cube}.

\begin{proof}
The main idea consists of modifying the first iteration of Algorithm~\ref{algo:meiser}, by
letting $\varepsilon = n^{-\frac{k}{2}}$.
Hence we pick a random subset $\net$ of
$O(n^{k/2 + 2} \log^2 n)$ hyperplanes in $\H$ and use this as an \enet.
This can be done efficiently, as shown in~Appendix~\ref{app:sampling}.

Next, we need to locate the input $q$ in the arrangement induced by
$\net$. This can be done by running Algorithm~\ref{algo:meiser} on the set
$\net$. From the previous considerations on Algorithm~\ref{algo:meiser}, the
running time of this step is
$$
O(n|\net|) = \tilde{O}(n^{k/2 + 4}),
$$
and the number of queries is $O(n^3\log^3n)$.

Then, in order to be able to prune the hyperplanes in $\H$, we have to
compute a simplex $\simplex$ that does not intersect any hyperplane of
$\net$. For this, we observe that the above call to Algorithm~\ref{algo:meiser}
involves computing a sequence of simplices for the successive pruning
steps. We save the description of those simplices. Recall that there are
$O(\log n)$ of them, all of them contain the input $q$ and have vertices
coinciding with vertices of the original arrangement $\arrangement(\H)$. In
order to compute a simplex $\simplex$ avoiding all hyperplanes of
$\net$, we can simply apply Algorithm~\ref{algo:simplex} on the set of
hyperplanes bounding the intersection of these simplices. The running time
and number of queries for this step are bounded respectively by
$n^{O(1)}$ and $O(n^2\log n)$.

Note that the vertices of $\simplex$ are not vertices
of $\arrangement(\H)$ anymore. However, their coordinates lie in a finite set
(see~Appendix~\ref{app:bound})
\begin{lemma}\label{lem:bound}
Vertices of $\simplex$ have rational coordinates whose fraction representations
have their numerators and denominators absolute value bounded by
$O(C^{4n^5} n^{4n^5+2n^3+n})$, where $C$ is a constant.
\end{lemma}

We now are in position to perform the pruning of the hyperplanes in $\H$ with
respect to $\simplex$. The number of remaining hyperplanes after the pruning is
at most $\varepsilon n^k = O(n^{k/2})$. Hence from Corollary~\ref{cor:double}, the
pruning can be performed in time proportional $\tilde{O}(n^{\ceil{k/2} + 7})$.

Similarly, we can detect any hyperplane of $\H$ containing $\simplex$ using the
result of Lemma~\ref{lem:multiple} in time $\tilde{O}(n^{\lceil k/2\rceil +
8})$. Note that those last two steps do not require any query.

Finally, it remains to detect any solution that may lie in the remaining
set of hyperplanes of size $O(n^{k/2})$. We can again fall back
on Algorithm~\ref{algo:meiser}, restricted to those hyperplanes. The running
time is $\tilde{O}(n^{k/2 + 2})$, and the number of queries is still $O(n^3 \log^3
n)$.

Overall, the maximum running time of a step is $\tilde{O}(n^{\ceil{\frac{k}{2}} + 8})$,
while the number of queries is always bounded by $O(n^3\log^3 n)$.

\end{proof}

\section{Query size}
\label{sec:query-size}

In this section, we consider a simple blocking scheme that allows us to explore
a tradeoff between the number of queries and the size of the queries.

\begin{lemma}
\label{lem:block}
For any integer $b>0$, an instance of the \kSUM{} problem on $n>b$ numbers can be split into
$O(b^{k-1})$ instances on at most $k\lceil \frac{n}{b}\rceil$ numbers, so that every $k$-tuple
forming a solution is found in exactly one of the subproblems.
The transformation can be carried out in time $O(n\log n + b^{k-1})$.
\end{lemma}
\begin{proof}
Given an instance on $n$ numbers, we can sort them in time $O(n\log n)$, then partition
the sorted sequence into \(b\) consecutive blocks \(B_1, B_2,\ldots ,B_b\) of equal size.
This partition can be associated with a partition of the real line
into $b$ intervals, say $I_1, I_2,\ldots ,I_b$. Now consider the partition of $\R^k$
into grid cells defined by the $k$th power of the partition $I_1, I_2,\ldots ,I_b$. The
hyperplane of equation $x_1 + x_2 +\cdots +x_k = 0$ hits $O(b^{k-1})$ such grid cells.
Each grid cell $I_{i_1}\times I_{i_2}\times \cdots \times I_{i_k}$ corresponds to a
\kSUM{} problem on the numbers in the set $B_{i_1}\cup B_{i_2}\cup \ldots \cup B_{i_k}$ (note that
the indices $i_j$ need not be distinct). Hence each such instance has size at most $k\lceil \frac{n}{b}\rceil$.
\end{proof}

Combining Lemma~\ref{lem:block} and Theorem~\ref{thm:cube} directly yields the following.
\begin{theorem}\label{thm:query-size}
For any integer $b>0$, there exists a $k\lceil \frac{n}{b}\rceil$-linear
decision tree of depth \(\tilde{O}(b^{k-4} n^3)\) solving the \kSUM{}
problem.
Moreover, this decision tree can be implemented as an
$\tilde{O}(b^{\floor{\frac{k}{2}}-9}n^{\ceil{\frac{k}{2}}+8})$
Las Vegas algorithm.
\end{theorem}

The following two corollaries are obtained by taking $b=\frac{k}{\varepsilon}$, and
$b=O(n^{\varepsilon})$, respectively
\begin{corollary}\label{cor:en}
	For any constant $\varepsilon>0$ such that $b=\frac{k}{\varepsilon}$ and
	$\frac{n}{b}$ are positive integers,
there exists an $\varepsilon n$-linear decision tree of
depth $\tilde{O} (n^3)$ solving the \kSUM{} problem.
Moreover, this decision tree can be implemented as an
$\tilde{O}(n^{\ceil{\frac{k}{2}}+8})$ Las Vegas algorithm.
\end{corollary}
\begin{corollary}\label{cor:ne}
For any $\varepsilon$ such that $0<\varepsilon<1$,
there exists an $O(n^{1-\varepsilon})$-linear decision tree of
depth $\tilde{O} (n^{3+(k-4)\varepsilon})$ solving the \kSUM{} problem.
Moreover, this decision tree can be implemented as an
$\tilde{O}(n^{(1+\varepsilon)\frac{k}{2} + 8.5})$
Las Vegas algorithm.
\end{corollary}

Note that the latter query complexity improves on $\tilde{O}(n^{\frac{k}{2}})$
whenever \(\varepsilon < \frac{1}{2}\) and $k\geq
\frac{3-4\varepsilon}{\sfrac{1}{2}-\varepsilon}$. Hence for instance, we obtain
an $O(n^\frac{3}{4})$-linear decision tree of depth $\tilde{O}(n^4)$ for the
\dsum[8] problem, and $o(n)$-linear decision trees of depth $o(n^4)$ for any
$k$.

\bibliography{bibliography}
\appendix

\section{Keeping queries linear in Algorithm~\ref*{algo:simplex}}
\label{app:keeplinear}
In Algorithm~\ref{algo:simplex}, we want to ensure that the queries we make in step
\step{2} are linear and that the queries we will make in the recursion step
remain linear too.
\begin{lemma}
	Algorithm~\ref{algo:simplex} can be implemented so that it uses $O(n\card{I})$ linear queries.
\end{lemma}%
\begin{proof}
Let us first analyze what the queries of step \step{2} look like. In addition
to the input point \(q\) we are given a vertex \(\nu\) and we want to find the
projection \(q'\) of \(q\) in direction \(\vec{\nu q}\) on the hyperplanes of
\(\I_{\theta}\). Let the equation of \(H_{i}\) be \(\Pi_{i}(x) =
c_{i} + d_{i}
\cdot x = 0\) where \(c_{i}\) is a scalar and \(d_{i}\) is a
vector.
The projection of \(q\) along \(\vec{\nu q}\) on a hyperplane \(H_i\) can thus
be written\footnote{Note that we project from \(\nu\) instead of \(q\). We are
allowed to do this since \(\nu + \lambda_{i} \vec{\nu q} = q + (\lambda_i - 1)
\vec{\nu q}\) and there is no hyperplane
separating \(q\) from \(\nu\).}
\(\rho(q,\nu,H_i) = \nu + \lambda_{i} \vec{\nu q}\) such that \(\Pi_{i}(\nu +
		\lambda_{i} \vec{\nu q}) = c_{i} + d_{i} \cdot \nu +
		\lambda_{i} d_{i} \cdot
\vec{\nu q} = 0\). Computing the closest hyperplane amounts to finding
\(\lambda_{\theta} = \min_{\lambda_i > 0} \lambda_i\). Since \(\lambda_i = -
\frac{c_i + d_i \cdot \nu}{d_i \cdot \vec{\nu q}}\) we can test whether
\(\lambda_i > 0\)
using the linear query\footnote{Note that if $c_i + d_i \cdot \nu = 0$ then
$\lambda_i=0$, we can check this beforehand for free.}
\(-\frac{d_i \cdot \vec{\nu q}}{c_i + d_i
\cdot \nu} \ask{>} 0\). Moreover, if \(\lambda_i > 0\) and \(\lambda_j > 0\)
we can test whether $\lambda_i < \lambda_j$ using the linear query \(
\frac{d_i \cdot \vec{\nu q}}{c_i + d_i \cdot \nu}
\ask{<}
\frac{d_j \cdot \vec{\nu q}}{c_j + d_j \cdot \nu}\).
Step \step{2} can thus be achieved using \(O(1)\) \((2k)\)-linear queries per
hyperplane of \(\net\).

In step \step{4}, the
recursive step is carried out on \(q' = \nu + \lambda_{\theta} \vec{\nu q} = \nu -
\frac{c_{\theta} + d_{\theta} \cdot \nu}{d_{\theta} \cdot \vec{\nu q}}
\vec{\nu q}\) hence comparing \(\lambda'_i\) to \(0\) amounts to performing the
query \(-\frac{d_i \cdot \vec{\nu q}'}{c_i + d_i \cdot \nu'}
\ask{>} 0\), which is not linear in \(q\). The same goes for comparing
\(\lambda'_i\) to \(\lambda'_j\) with the query
\(\frac{d_i \cdot \vec{\nu q}'}{c_i + d_i \cdot \nu'}
\ask{<}
\frac{d_j \cdot \vec{\nu q}'}{c_j + d_j \cdot \nu'}\).

However, we can multiply both sides of the inequality test by \(d_\theta
\vec{\nu q}\) to keep the queries linear as shown below. We must be careful to
take into account the sign of the expression \(d_\theta \vec{\nu q}\), this
costs us one additional linear query.

This trick can be used at each step of the recursion. Let \(q^{(0)} = q\),
then we have
$$
	q^{(s+1)} = \nu^{(s)} - \frac{c_{\theta_{s}} + d_{\theta_{s}} \cdot
	\nu^{(s)}}{d_{\theta_{s}} \cdot \vec{\nu q}^{(s)}}\vec{\nu q}^{(s)}
$$
and \( (d_{\theta_{s}}\cdot \vec{\nu q}^{(s)}) q^{(s+1)}\) yields a vector
whose components are linear in \(q^{(s)}\).
Hence,
\( (\prod_{k=0}^{s} d_{\theta_{k}} \cdot \vec{\nu q}^{(k)})
 q^{(s+1)}\) yields a vector
whose components are linear in \(q\),
and for all pairs of vectors \(d_i\) and \(\nu^{(s+1)}\)
we have that \( (\prod_{k=0}^{s} d_{\theta_{k}} \cdot \vec{\nu q}^{(k)}) (d_i
\cdot \vec{\nu q}^{(s+1)})\) is linear in \(q\).

Hence at the $s$th recursive step of the algorithm, we will perform
at most \(\card{\net}\) linear queries of the type
$$
	-\left(\prod_{k=0}^{s-1} d_{\theta_{k}} \cdot \vec{\nu q}^{(k)}\right)
\frac{d_i \cdot \vec{\nu q}^{(s)}}{c_i + d_i \cdot
	\nu^{(s)}} \ask{>} 0
$$
\(\card{\net} - 1\) linear queries of the type
$$
	\left(\prod_{k=0}^{s-1} d_{\theta_{k}} \cdot \vec{\nu q}^{(k)}\right)
	\frac{d_i \cdot \vec{\nu q}^{(s)}}{c_i + d_i \cdot \nu^{(s)}}
\ask{<}
\left(\prod_{k=0}^{s-1} d_{\theta_{k}} \cdot \vec{\nu q}^{(k)}\right)
\frac{d_j \cdot \vec{\nu q}^{(s)}}{c_j + d_j \cdot \nu^{(s)}}
$$
and a single linear query of the type
$$
	d_{\theta_{s-1}} \cdot \vec{\nu q}^{(s-1)} \ask{<} 0.
$$

In order to detect all hyperplanes \(H_i\) such that \(\lambda_i =
\lambda_\theta\) we can afford to compute the query $f(q) > g(q)$ for all query
$f(q) < g(q)$ that we issue, and vice versa.

Note that, without further analysis, the queries can become \(n\)-linear as
soon as we enter the \(\frac{n}{k}^{\text{th}}\) recursive step.
\end{proof}

\section{Algebraic computation trees}
\label{app:act}

We consider \emph{algebraic computation trees}, whose internal nodes
are labeled with arithmetic (\(r \gets o_1 \op o_2, \op \in
\{\,+,-,\times,\div\,\}\)) and branching (\(z : 0\)) operations. We say that an
algebraic computation tree $T$ \emph{realizes} an algorithm $A$ if the paths
from the root to the leaves of $T$ correspond to the execution paths of \(A\)
on all possible inputs \(q \in \mathbb{R}^n\), where \(n\) is fixed. A leaf is
labeled with the output of the corresponding execution path of \(A\). Such a
tree is \emph{well-defined} if any internal node labeled \(r \gets o_1 \op
o_2\) has outdegree \(1\) and is such that either \(o_k = q_i\) for some \(i\)
or there exists an ancestor \(o_k \gets x \op y\) of this node, and any
internal node labeled \(z : 0\) has outdegree \(3\) and is such that either \(z
= q_i\) for some \(i\) or there exists an ancestor \(z \gets x \op y\) of this
node. In the algebraic computation tree model, we define the complexity
\(f(n)\) of an algorithm $A$ to be the minimum depth of a well-defined
computation tree that realizes $A$ for inputs of size $n$.

In the algebraic computation tree model, we only count the operations that
involve the input, that is, members of the input or results of previous
operations involving the input. The following theorem follows immediately from
the analysis of the linearity of queries
\begin{theorem}\label{thm:act}
	The algebraic computation tree complexity of \(k\)-LDT is
	\(\tilde{O}(n^3)\).
\end{theorem}

\begin{proof}
We go through each step of Algorithm~\ref{algo:meiser}.
Indeed, each \(k\)-linear query of step \step{1} can be implemented as
\(O(k)\) arithmetic operations, so step \step{1} has complexity
\(O(\card{\net})\).
The construction of the simplex in step \step{2} must be handled carefully.
What we need to show is that each \(n\)-linear query we use can be implemented
using $O(k)$ arithmetic operations. It is not difficult to see from the
expressions given in~Appendix~\ref{app:keeplinear} that a constant number of arithmetic
operations and dot products suffice to
compute the queries. A dot product in this case involves a constant number
of arithmetic operations because the \(d_i\) are such that they each have
exactly \(k\) non-zero components. The only expression that involves a
non-constant number of operations is the product \(\prod_{k=0}^{s}
d_{\theta_{k}} \cdot \vec{\nu q}^{(k)}\), but this is equivalent to
\((\prod_{k=0}^{s-1}
	d_{\theta_{k}} \cdot \vec{\nu q}^{(k)})(d_{\theta_{s}} \cdot
	\vec{\nu q}^{(s)})\)
where the first factor has already been computed during a previous step and
the second factor is of constant complexity. Since each query costs a constant
number of arithmetic operations and branching operations, step \step{2}
has complexity \(O(n\card{\net})\).
Finally, steps \step{3} and \step{4} are free since they do not involve the
input. The complexity of Algorithm~\ref{algo:meiser} in this model is thus also \(O(n^3
\log^2 n \log \card{\H{}})\).
\end{proof}

\section{Uniform random sampling}
\label{app:sampling}

Theorem~\ref{thm:enet} requires us to pick a sample of the hyperplanes
uniformly at random. Actually the theorem is a little stronger; we can draw
each element of $\net$ uniformly at random, only keeping distinct elements.
This is not too difficult to achieve for $k$-LDT when the $\alpha_i,i\in[k]$
are all distinct: to pick a hyperplane of the form $\alpha_0 + \alpha_1 x_{i_1}
+ \alpha_2 x_{i_2} + \cdots + \alpha_k x_{i_k} = 0$ uniformly at random, we can
draw each $i_j\in[n]$ independently and there are $n^k$ possible outcomes.
However, in the case of $k$-SUM, we only have $\binom{n}{k}$ distinct
hyperplanes. A simple dynamic programming approach solves the problem for
\kSUM{}. For \kLDT{} we can use the same approach, once for each class of equal
$\alpha_i$.

\begin{lemma}
	Given $n\in\N$ and $(\alpha_0,\alpha_1,\ldots,\alpha_k) \in \R^{k+1}$,
	a uniform random sample $\net$ of hyperplanes in $\R^n$ with equations of
	the form $\alpha_0 + \alpha_1 x_{i_1} + \alpha_2 x_{i_2} + \cdots +
	\alpha_k x_{i_k} =
	0$ can be computed in time $O(k\card{\net})$ and preprocessing time $O(kn)$.
\end{lemma}
\begin{proof}
	We want to pick
	an assignment
	$a=\enum{(\alpha_1,x_{i_1}),(\alpha_2,x_{i_2}),\ldots,(\alpha_k,x_{i_k})}$
	uniformly at random. Note that all $x_i$ are distinct while the $\alpha_j$
	can be equal.

	Without loss of generality, suppose $\alpha_1 \le \alpha_2 \le \cdots \le
	\alpha_k$. There is a bijection between assignments and lexicographically
	sorted $k$-tuples
	$((\alpha_1,x_{i_1}),(\alpha_2,x_{i_2}),\ldots,(\alpha_k,x_{i_k}))$.

	Observe that $x_{i_j}$ can be drawn independently of $x_{i_{j'}}$ whenever
	$\alpha_j \neq \alpha_{j'}$. Hence, it suffices to generate a lexicographically sorted
	$\card{\chi}$-tuple of $x_i$ for each class $\chi$ of equal $\alpha_i$.

	Let $\omega(m,l)$ denote the number of lexicographically sorted $l$-tuples,
	where each element comes from a set of $m$ distinct $x_i$.
	We have
	$$
		\omega(m,l) = \begin{cases}
			1 & \text{if}\ l = 0\\
			\sum_{i=1}^{m}\omega(i,l-1) & \text{otherwise.}
		\end{cases}
	$$

	To pick such a tuple $(x_{i_1},x_{i_2},\ldots,x_{i_l})$ uniformly at random
	we choose $x_{i_l} = x_o$ with probability
	$$
		P(x_{i_l} = x_o) = \begin{cases}
			0 & \text{if}\ o > m\\
			\frac{\omega(o,l-1)}{\omega(m,l)} & \text{otherwise}
		\end{cases}
	$$
	that we append to a prefix $(l-1)$-tuple (apply the procedure recursively),
	whose elements come from a set of $o$ symbols. If $l=0$ we just
	return the empty tuple.

	Obviously, the probability for a given $l$-tuple to be picked is equal to
	$\frac{1}{\omega(m,l)}$.

	Let $X$ denote the partiton of the $\alpha_i$ into equivalence classes,
	then the number of
	assignments is equal to $\prod_{\chi \in X}\omega(n,\card{\chi})$.
	(Note that for $k$-SUM this is simply $\omega(n,k)$ since there is only a
	single class of equivalence.)
	For each equivalence class $\chi$ we draw independently a lexicographically
	sorted $\card{\chi}$-tuple on $n$ symbols using the procedure
	above. This yields a given assignment with probability
	$\frac{1}{\prod_{\chi \in X}\omega(n,\card{\chi})}$.
	Hence, this corresponds to a uniform random draw over the assignments.

	For given $n$ and $k$, there are at most $nk$ values $\omega(m,l)$ to
	compute, and for a given $k$-LDT instance, it must be computed only once.
	Once those values have been computed, making a random draw takes time
	$O(k)$.

\end{proof}

\section{Proof of Lemma~\ref*{lem:bound}}
\label{app:bound}
\begin{theorem}[Cramer's rule]\label{thm:cramer}
	If a system of $n$ linear equations for $n$ unknowns, represented in matrix
	multiplication form $Ax=b$,
	has a unique solution $x=(x_1,x_2,\ldots,x_n)^T$ then, for all $i \in [n]$,
	$$
		x_i = \frac{\det(A_i)}{\det(A)}
	$$
	where $A_i$ is $A$ with the $i$th column replaced by the column vector $b$.
\end{theorem}

\begin{lemma}\label{lem:detZ}
	The absolute value of the determinant of an $n\times n$ matrix $M$ with
	integer entries is $O((Cn)^n)$ where $C$ is the maximum absolute value in
	$M$.
\end{lemma}
\begin{proof}
	The determinant of such a matrix is the sum of $2n!$ integers with maximum
	absolute value $C^n$.
\end{proof}

\begin{lemma}\label{lem:detQ}
	The determinant of an $n\times n$ matrix $M$ with
	rational entries can be represented as a fraction whose numerators and
	denominators absolute values are bounded by $O((ND^{n-1}n)^n)$ and $O(D^{n^2})$
	respectively, where $N$ and $D$
	are respectively the maximum absolute value of a numerator and a
	denominator.
\end{lemma}
\begin{proof}
	Multiply each row $M_i$ of $M$ by $\prod_j d_{i,j}$. Apply Lemma~\ref{lem:detZ}.
\end{proof}

We can now proceed to the proof of Lemma~\ref{lem:bound}.
\begin{proof}
	Coefficients of the hyperplanes of the arrangement are constant rational
	numbers, those can be changed to constant integers (because each
	hyperplane has at most $k$ nonzero coefficients). Let $C$ denote the
	maximum absolute value of those coefficients.

	Because of Theorem~\ref{thm:cramer} and Lemma~\ref{lem:detZ}, vertices of
	the arrangement have rational coordinates whose numerators and
	denominators absolute values are bounded by $O(C^n n^n)$.

	Given simplices whose vertices are vertices of the arrangement, hyperplanes
	that define the faces of those simplices have rational coefficients whose
	numerator and denominator are bounded by $O(C^{2n^3} n^{2n^3+n})$ by
	Theorem~\ref{thm:cramer} and Lemma~\ref{lem:detQ}. (Note that some
	simplices might be not fully dimensional, but we can handle those by adding
	vertices with coordinates that are not much larger than that of already
	existing vertices).

	By applying Theorem~\ref{thm:cramer} and Lemma~\ref{lem:detQ} again, we
	obtain that vertices of the arrangement of those new hyperplanes (and thus
	vertices of $\simplex$) have rational coefficients whose numerator and
	denominator are bounded by $O(C^{4n^5} n^{4n^5+2n^3+n})$.
\end{proof}

\end{document}